\documentclass[12pt]{article}
\usepackage[utf8]{inputenc}
\usepackage[english]{babel}
\usepackage{hyperref}
\usepackage{textcomp}
\usepackage[autostyle]{csquotes}
\usepackage{comment}
\usepackage[backend=biber,maxbibnames=99]{biblatex}
\usepackage{authblk}
\usepackage{amsthm}
\usepackage{amsfonts}
\usepackage{mathrsfs}
\usepackage{amssymb}
\usepackage{graphicx}
\usepackage{comment}
\usepackage{geometry}
\usepackage{subfiles}
\usepackage{subfig}
\usepackage{mathtools}
\usepackage{pgfplots}
\pgfplotsset{/pgf/number format/use comma,compat=newest}
\usepackage{url}
\usepackage{tikz}
\usepackage{bbm}
\usepackage{floatflt}
\usetikzlibrary{arrows,decorations.pathmorphing,backgrounds,positioning,fit,petri}

\newcommand\iid{\mathrel{\stackrel{\makebox[0pt]{\mbox{\normalfont\tiny iid}}}{\sim}}}

\theoremstyle{plain} 
\newtheorem{theorem}{Theorem}

\newtheorem{lemma}[theorem]{Lemma} 
\newtheorem{proposition}[theorem]{Proposition} 

\theoremstyle{definition} 
\newtheorem{definition}{Definition}

\theoremstyle{remark} 
\newtheorem{remark}{Remark}

\title{The multi-species mean-field spin-glass on the Nishimori line}
\author[1]{Diego Alberici}
\author[2]{Francesco Camilli}
\author[2]{Pierluigi Contucci}
\author[2]{Emanuele Mingione}
\affil[1]{Communication Theory Laboratory, EPFL, Switzerland}
\affil[2]{Dipartimento di Matematica, Università di Bologna, Italy}

\begin{document}





\maketitle

\begin{abstract}
In this paper we study a multi-species disordered model on the Nishimori line. The typical properties of this line, a set of identities and inequalities among correlation functions, allow us to prove the replica symmetry i.e. the concentration of the order parameter. When the interaction structure is elliptic we rigorously compute the exact solution of the model in terms of a finite-dimensional variational principle and we study its properties.
\end{abstract}

\textbf{keywords}: Multi-species spin glass, Nishimori line, replica symmetry\\

\section{Introduction}
In this paper we investigate the properties of the elliptic multi-species Sherrington-Kirkpatrick model along the Nishimori line i.e. the sub-manifold of the phase space in which mean and variance of the random parameters, interactions and magnetic fields, coincide.
The multi-species version of a mean field model is simply obtained by relaxing the full invariance under the symmetric group into the weaker one of the product of the symmetric groups on a given partition of the system. The ratios of the partition with respect to the whole, the form factors, are kept fixed in the thermodynamic limit. The ellipticity condition provides the positivity and monotonicity properties that allow to study the system with interpolation methods \cite{Guerra_upper_bound,interp_guerra_2002,panchenko2015sherrington} and obtain a Parisi like solution for Gaussian centered interactions and deterministic magnetic fields \cite{MSK_original,panchenko_multi-SK} (see also \cite{Chen_ferromagnetic} for a case with a ferromagnetic mean of the interactions).

The choice to study the model on the Nishimori line  \cite{nishimori01} reflects the importance of this sub-manifold of the phase space due to its ubiquitous appearence in inference problems and, especially, on the statistical physics approach to machine learning \cite{Lenka,wigner-wishart}. 

The main results of the paper, Theorem \ref{main_theorem} and Lemma \ref{concentration_lemma} in Section \ref{section_4}, are the proof of the variational expression for the pressure per particle in the thermodynamic limit and the self-averaging of the magnetization per particle. The techniques we use to prove them are obtained by merging methods whose origins belong both to statistical mechanics and high dimensional inference \cite{Albanese,Agliari_2020,adaptive,GG_contucci_Giardina,contucci_giardina_2012,GG_original,Guerra_upper_bound,contucci_morita_nishimori,morita2005}. 

The paper is organized as follows. In Section \ref{section_2} we give the definition of the model together with its main properties, such as the self-averaging of the pressure and the Nishimori identities. In Section \ref{section_3} we extend to our multi-dimensional model the adaptive interpolation method due to Barbier and Macris \cite{adaptive} and we use it to compute the exact solution in Section \ref{section_4} by writing the pressure in the thermodynamic limit in terms of a finite-dimensional variational principle. Finally we study the main properties of the extremizers of our variational expression. The conclusions summarise the results and specify the connection of our model with an inference problem of Wigner spiked type \cite{Lenka,Jean-lenka2}. In the Appendix \ref{Appendix_A} the reader can find the details of the proof of the concentration of the magnetization in the thermodynamic limit, which ultimately leads to replica symmetry. For completeness the properties of the mono-species case (SK) on the Nishimori line are studied in Appendix \ref{Appendix_B}.

\section{Definitions and basic properties}\label{section_2}
Consider a set $\Lambda$  of indices with cardinality $|\Lambda|=N$. Let us partition $\Lambda$ in $K$ disjoint subsets:
\begin{align}
    \Lambda=\bigcup_{r=1}^K\Lambda_r,\;\quad\Lambda_r\cap\Lambda_s=\emptyset\;\forall r\neq s,\;\quad |\Lambda_r|=:N_r,\quad\alpha_r:=\frac{N_r}{N}\in(0,1)
\end{align}
Each subset will be called \emph{species} from now on. The model is defined by the following Gaussian Hamiltonian:
\begin{align}
    \label{H_MSK_NL}
    &H_N(\sigma):=-\sum_{r,s=1}^K\sum_{(i,j)\in\Lambda_r\times\Lambda_s}\tilde{J}_{ij}^{rs}\sigma_i\sigma_j
    -\sum_{r=1}^K\sum_{i\in\Lambda_r}\tilde{h}^r_i\sigma_i
    ,\\
    \label{gaussian_couplings_NL}
    &\tilde{J}_{ij}^{rs}\iid\mathcal{N}\left(\frac{\mu_{rs}}{2N},\frac{\mu_{rs}}{2N}\right),\;\quad \tilde{h}^r_i\iid\mathcal{N}(h_r,h_r)
\end{align}where $\mu_{rs}$ and $h_r$ are positive real numbers, and the $K\times K$ matrix $\mu=(\mu_{rs})_{r,s=1,\dots,K}$ can be assumed to be symmetric without loss of generality. Throughout this work, as can be seen from the previous definitions, the family of Gaussian variables \eqref{gaussian_couplings_NL} are assumed to be in a special line where mean values and variances are tied to be identical. One can see that this condition, in the context of statistical mechanics, is known as Nishimori line and was introduced in \cite{Hide_original} for the SK model with Bernoulli couplings. For the Gaussian SK at inverse temperature $\beta$ and random couplings $J_{ij}\iid\mathcal{N}\left(\frac{J_0}{2N},\frac{J}{2N}\right)$ the Nishimori line is defined by $\beta J= J_0$ (see Paragraph 4.3 in \cite{nishimori01}) which is equivalent to \eqref{gaussian_couplings_NL} when $K=1$, which explains also why we set $\beta=1$ throughout the paper without loss of generality. We will see that on the Nishimori line a special set of identities and inequalities hold. 

It is also convenient to rewrite the Hamiltonian \eqref{H_MSK_NL} in terms of centered Gaussians. To do that we introduce the following notation for species magnetizations and overlaps that will be used throughout:
\begin{align}
    &m_r(\sigma):=\frac{1}{N_r}\sum_{i\in\Lambda_r}\sigma_i,\;\quad q_r(\sigma,\tau):=\frac{1}{N_r}\sum_{i\in\Lambda_r}\sigma_i\tau_i\\
    &\mathbf{m}(\sigma):=(m_r(\sigma))_{r=1,\dots,K},\quad
    \mathbf{q}(\sigma,\tau):=(q_r(\sigma,\tau))_{r=1,\dots,K}
\end{align}where bold characters here and below stand for vectors and $\sigma,\tau\in\Sigma_N:=\{-1,1\}^N$. We also set: 
\begin{align}
    \Delta:=(\alpha_r\mu_{rs}\alpha_s)_{r,s=1,\dots,K}\,,\;\quad\hat{\alpha}:=\text{diag}(\alpha_1,\alpha_2,\dots,\alpha_K)\,,\;\quad
    \mathbf{h}:=
    (h_r)_{r=1,\dots,K}\;
    .
\end{align}We will call $\Delta$ the \emph{effective interaction matrix} because it encodes the interactions and relative sizes of the species in our model and we notice that it is positive definite if and only if $\mu$ is. See \figurename\ref{main_fig} for a scheme.
\begin{figure}[h!]
\centering
\includegraphics[width=.9\textwidth]{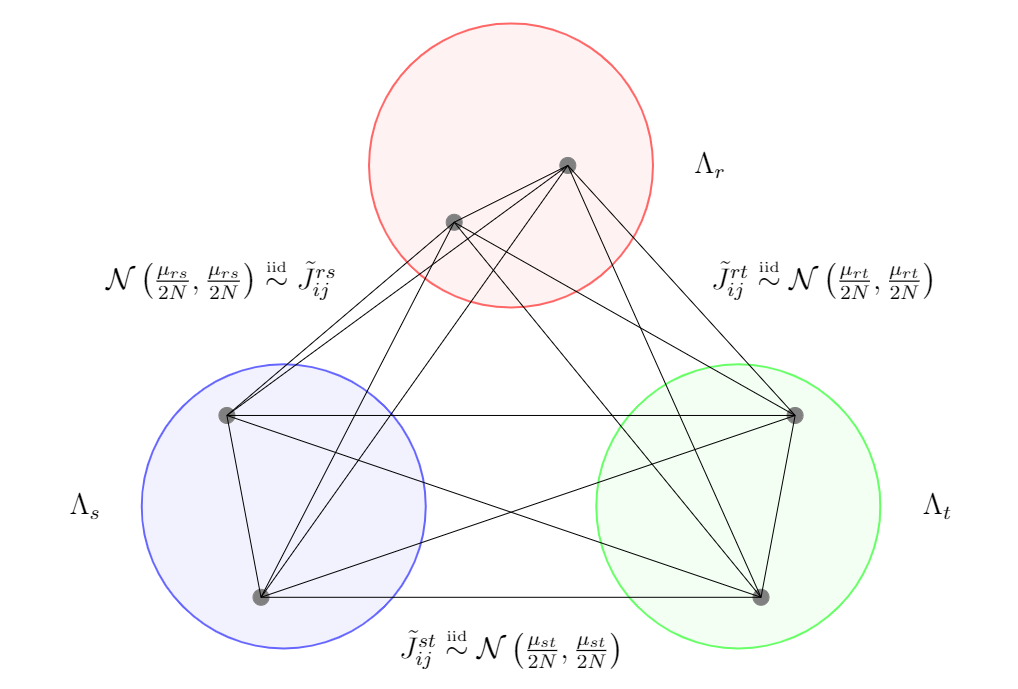}
\caption{Scheme of the structure of the interactions.}\label{main_fig}
\end{figure}

With these notations we can write a Hamiltonian in terms of centered Gaussian variables which is equivalent in distribution to the one in \eqref{H_MSK_NL}:
\begin{multline}
    \label{H_MSK_NL-bis}
    H_N(\sigma)=-\frac{1}{\sqrt{2N}}\sum_{r,s=1}^K\sum_{(i,j)\in\Lambda_r\times\Lambda_s}J_{ij}^{rs}\sigma_i\sigma_j-\sum_{r=1}^K\sum_{i\in\Lambda_r}h^r_i\sigma_i+\\-\frac{N}{2}(\mathbf{m},\Delta\mathbf{m})-N(\hat{\alpha}\mathbf{h},\mathbf{m})\,,\quad
    J_{ij}^{rs}\iid\mathcal{N}\left(0,\mu_{rs}\right),\;\quad
    h_i^r\iid\mathcal{N}(0,h_r)\;.
\end{multline}
The last expression allows us to identify the model with a multi-species Sherrington-Kirkpatrick model (SK) with the addition of a ferromagnetic interaction and a positive external field whose intensity coincide with the variances of the random terms. 

Now we define the main quantity under investigation, the random and average quenched pressure densities:
\begin{align}
    \label{press}
    &p_N:=\frac{1}{N}\log\sum_{\sigma\in\Sigma_N}\exp\left(-H_N(\sigma)
    \right)\\ 
    \label{Q_press}
    &\bar{p}_N(\mu,h):=\mathbb{E}p_N 
\end{align}
where we emphasize the dependence of the quenched pressure on the mean parameters $\mu_{rs},\, h$ and the symbol $\mathbb{E}$ stands for the Gaussian expectation with respect to the disorder.
We also introduce the Gibbs expectation:
\begin{align}
    \langle\cdot\rangle_N:=\frac{\sum_{\sigma\in\Sigma_N}e^{-H_N(\sigma)}(\cdot)}{Z_N},\,\quad Z_N:=\sum_{\sigma\in\Sigma_N}e^{-H_N(\sigma)}
\end{align}
We will denote the dependence of the Gibbs measure on further parameters with subscripts or superscripts, for example $\langle\cdot\rangle_{N,t\dots}^{(\epsilon)}$. Notice that in this context the Gibbs measure is random.

The following concentration property for the pressure density holds true. It will be an important tool to prove replica symmetry when combined with the Nishimori identities introduced in the next section.

\begin{proposition}\label{tala_concentration}
There exists $C=C(\mu,h)>0$ such that for every $x>0$
\begin{align} \label{tala_bound}
    \mathbb{P}\left(\left|p_N-\bar{p}_N(\mu,h)\right|\geq x\right) \,\leq\, 2\exp\left(-\frac{Nx^2}{4C}\right) \;.
\end{align}
As a consequence
\begin{align} \label{self_av_press}
    \mathbb{E}[(p_N-\bar{p}_N(\mu,h))^2] \,\leq\, \frac{8C}{N}  \;.
\end{align}
\end{proposition}

\begin{proof}
The random pressure $p_N$ is a Lipschitz function of the independent standard Gaussian variables $\hat J =(J_{ij}^{rs}/\sqrt{\mu_{rs}})_{i,j,r,s}\,$, $\hat h=(h_i^r/\sqrt{h_r})_{i,r}\,$. Indeed:
\begin{align}
   N^2\,\Vert\nabla_{\!\hat J,\hat h}\,p_N\Vert^2 \leq N\left(\frac{(\mathbf{1},\Delta\mathbf{1})}{2}+(\hat{\alpha}\mathbf{h},\mathbf{1})\right)\equiv\, C\,N
\end{align}
The inequality \eqref{tala_bound} then follows by a standard concentration property of the Gaussian measure (see Theorem 1.3.4 in \cite{Tala_vol1}). A tail integration finally leads to \eqref{self_av_press}.
\end{proof}

\subsection{Nishimori identities and correlation inequalities}
Here we will list some identities and inequalities on the Nishimori line. The identities were introduced in the original work by H. Nishimori \cite{Hide_original}, while the inequalities were noticed and proved much later \cite{contucci_morita_nishimori,morita2005}. The proof of the Nishimori identities that is most suitable for our model can be found in Paragraph 2.6 of \cite{contucci_giardina_2012}. In particular, for our purposes, we will need
\begin{align}
    \label{N_identity_1}
    &\mathbb{E}[\langle\sigma_i\rangle_N^2]=\mathbb{E}[\langle\sigma_i\rangle_N]\\
    \label{N_identity_2}
    &\mathbb{E}[\langle\sigma_i\sigma_j\rangle_N^2]=\mathbb{E}[\langle\sigma_i\sigma_j\rangle_N]\\
    \label{N_identity_3}
    &\mathbb{E}[\langle\sigma_i\rangle_N\langle\sigma_i\sigma_j\rangle_N]=\mathbb{E}[\langle\sigma_i\rangle_N\langle\sigma_j\rangle_N]
\end{align}
for all $i,j\in\Lambda$.
In particular they imply that:
\begin{multline}
    \mathbb{E}[\langle q_s\rangle_N]=\sum_{i\in\Lambda_s}\frac{1}{N_s}\mathbb{E}[\langle\sigma_i\rangle_N\langle\tau_i\rangle_N]=
    \sum_{i\in\Lambda_s}\frac{1}{N_s}\mathbb{E}[\langle\sigma_i\rangle_N^2]= \sum_{i\in\Lambda_s}\frac{1}{N_s}\mathbb{E}[\langle\sigma_i\rangle_N]=\\=\mathbb{E}[\langle m_s\rangle_N]
\end{multline}
\begin{multline}
    \mathbb{E}[\langle q_rq_s\rangle_N]=\sum_{(i,j)\in\Lambda_r\times\Lambda_s}\frac{\mathbb{E}[\langle\sigma_i\sigma_j\rangle_N\langle\tau_i\tau_j\rangle_N]}{N_rN_s}=
    \sum_{(i,j)\in\Lambda_r\times\Lambda_s}\frac{\mathbb{E}[\langle\sigma_i\sigma_j\rangle_N^2]}{N_rN_s}=\\=\mathbb{E}[\langle m_rm_s\rangle_N]
\end{multline}
    
and finally:
\begin{align}
\label{main_N_identity_quadratic}
    \mathbb{E}\Big\langle
    (\mathbf{q},\Delta\mathbf{q})
    \Big\rangle_N=\mathbb{E}\Big\langle
    (\mathbf{m},\Delta\mathbf{m})
    \Big\rangle_N\,.
\end{align}
The previous identities show that the model has a unique order parameter, that can be regarded as a magnetization or equivalently an overlap. We will choose the first point of view. This intuitive statement will acquire a precise meaning when we will write down the sum rule for the quenched pressure.

Following \cite{Nishi_id_PC,contucci_morita_nishimori,morita2005} (see  Theorem 2.18 in \cite{contucci_giardina_2012} for a straightforward proof) we obtain the I and II type correlation inequalities respectively:
\begin{align}
    \label{corr_in_1}
    &\frac{\partial \bar{p}_N}{\partial h_r}=
    \frac{1}{2N}\sum_{i\in\Lambda_r}\mathbb{E}[1+\langle\sigma_i\rangle_N]=\frac{\alpha_r}{2}[1+\mathbb{E}\langle m_r\rangle_N]\geq 0\\
    \label{corr_in_2}
    &\frac{\partial^2 \bar{p}_N}{\partial h_r\partial h_s}=\frac{\alpha_r}{2}\frac{\partial}{\partial h_s}\mathbb{E}\langle m_r\rangle_N=
    \frac{1}{2N}\sum_{(i,j)\in\Lambda_r\times\Lambda_s}
    \mathbb{E}[(\langle\sigma_i\sigma_j\rangle_N-\langle\sigma_i\rangle_N\langle\sigma_j\rangle_N)^2]\geq 0\,.
\end{align}
Analogous identities and inequalities hold for the first and second derivatives w.r.t. $\mu_{rs}$. The pressure and the first moment are monotonically increasing with respect to the Nishimori parameters $\mu_{rs}$, $h_r$. In particular the magnetization is always increasing w.r.t. the external field mean:
\begin{align}
    \label{mag_increasign_h}
    \frac{\partial\mathbb{E}\langle m_r\rangle_N}{\partial h_s}\geq 0
\end{align}
This monotonicity will be a key ingredient to prove replica symmetry.

\section{Adaptive interpolation and sum rule}\label{section_3}
In this section we build up an interpolating model with some specific features. The method here employed is an extension of the standard Guerra-Toninelli interpolation \cite{interp_guerra_2002}, also called \emph{adaptive interpolation technique}, developed in \cite{adaptive} by J. Barbier and N. Macris. 

\begin{definition}[Interpolating model] Let $t\in[0,1]$. The hamiltonian of the interpolating model is:
\begin{multline}
    \label{H_interpolating_model}
    H_\sigma(t):=-\frac{\sqrt{1-t}}{\sqrt{2N}}\sum_{r,s=1}^K\sum_{(i,j)\in\Lambda_r\times\Lambda_s}J_{ij}^{rs}\sigma_i\sigma_j-(1-t)\frac{N}{2}(
    \mathbf{m},\Delta\mathbf{m})+\\-
    \sum_{r=1}^K\sum_{i\in\Lambda_r}\left(\sqrt{Q_{\epsilon,r}(t)}J_i^r+Q_{\epsilon,r}(t)\right)\sigma_i-\sum_{r=1}^K\sum_{i\in\Lambda_r}h^r_i\sigma_i-N(\hat{\alpha}\mathbf{h},\mathbf{m})
\end{multline}
with $J_{i}^{r}\iid\mathcal{N}\left(0,1\right)$ independent of all the other Gaussian random variables, and
\begin{align*}
    \mathbf{Q}_{\epsilon}(t):=\boldsymbol{\epsilon}+\hat{\alpha}^{-1}\Delta \int_0^t\mathbf{q}_\epsilon(s)\,ds,\;\quad\epsilon_r\in[s_N,2s_N],\,s_N\propto N^{-\frac{1}{16K}} \;.
\end{align*}
Here $\mathbf{Q}_{\epsilon}=:(Q_{\epsilon,r})_{r=1,\dots,K}$, while $\mathbf{q}_\epsilon:=(q_{\epsilon,r})_{r=1,\dots,K}$ denotes a vector of $K$ non-negative functions that will be suitably chosen in the following.
\end{definition}

\begin{remark}
We notice that the interpolating model is on the Nishimori line for any $t\in [0,1]$. 
 Indeed 
\eqref{H_interpolating_model} equals in distribution the following Hamiltonian 
\begin{equation}\label{H_EMANUE}
     \tilde{H}_\sigma(t)\,=-\,\sum_{r,s=1}^K\sum_{(i,j)\in\Lambda_r\times\Lambda_s}\tilde{J}_{ij}^{rs}(t)\sigma_i\sigma_j\,-\,
    \sum_{r=1}^K\sum_{i\in\Lambda_r} \tilde{J}_i^{\epsilon,r}(t)\sigma_i\,-\,\sum_{r=1}^K\sum_{i\in\Lambda_r}\tilde{h}^r_i\sigma_i
\end{equation}
where 
\begin{equation}
\tilde{J}_{ij}^{rs}(t)\iid\mathcal{N}\left(\frac{(1-t)\mu_{rs}}{2N},\frac{(1-t)\mu_{rs}}{2N}\right),\;\quad\tilde{J}_i^{\epsilon,r}(t)\iid\mathcal{N}\left(Q_{\epsilon,r}(t),Q_{\epsilon,r}(t)\right)
\end{equation}
and  $\tilde{h}^r_i$ is defined in \eqref{gaussian_couplings_NL}. Given $t\in[0,1]$,   $\tilde{H}_\sigma(t)$ is a linear combination of independent non-centered Gaussian random variables where mean equals variance.  Therefore the Nishimori identities \eqref{N_identity_1}, \eqref{N_identity_2} and \eqref{main_N_identity_quadratic} can be used by replacing $\langle\cdot\rangle$ with the Gibbs measure induced by the interpolating hamiltonian \eqref{H_interpolating_model}, that is $\langle\cdot\rangle_{N,t}^{(\epsilon)}$.
Notice also that the role played by the functions $\mathbf{Q}_{\epsilon}(t)$ is that of an external magnetic field.
\end{remark}

The corresponding interpolating pressure will be denoted as
\begin{align}\label{interp_pressure}
    \bar{p}_{N,\epsilon}(t):=\frac{1}{N}\mathbb{E}\log\sum_\sigma e^{-H_\sigma(t)}\,.
\end{align}
In the previous equation and in the rest of the paper we drop the explicit dependence on $\mathbf{q}_\epsilon(t)$ to lighten the notation.

The following lemma will lead to the sum rule of the model.
\begin{lemma}[Interpolating pressure at $t=0,1$]\label{pressione_estremi}
Setting
\begin{align}
    \label{gas_pressure}
    \psi(Q):=\mathbb{E}_z\log2\cosh\left[z\sqrt{Q}+Q\right],\;\quad z\sim\mathcal{N}(0,1)
\end{align}
we have the following:
\begin{align}
    \label{extreme_1_press}
    \begin{split}
    \bar{p}_{N,\epsilon}(1) \,&=\, \sum_{r=1}^K\alpha_r\,\psi(Q_{\epsilon,r}(1)+h_r) \,=\\
    &=\mathcal{O}(s_N)+\sum_{r=1}^K\alpha_r\psi\left(\left(\hat{\alpha}^{-1}\Delta \int_0^1\mathbf{q}_\epsilon(t)\,dt+\mathbf{h}\right)_{\!r}\right)\end{split}\\
    \label{extreme_0_press}
    \bar{p}_{N,\epsilon}(0) \,&=\, \mathcal{O}(s_N)+\bar{p}_N(\mu,h)\,.
\end{align}
\end{lemma}
\begin{proof} 
Each $\epsilon_r$ can be regarded as the mean (or variance) of a small magnetic field.

At $t=1$ the system is \emph{free}, non interacting. Its pressure can be explicitly computed. Take $z_i^r\iid\mathcal{N}(0,1)$. Then:
\begin{multline*}
    \bar{p}_{N,\epsilon}(1)=\frac{1}{N}\mathbb{E}\log\prod_{r=1}^K\sum_{\sigma\in\Sigma_{N_r}}\exp\left(
    \sum_{i\in\Lambda_r}\left(\sqrt{Q_{\epsilon,r}(1)}J_i^r+Q_{\epsilon,r}(1)\right)\sigma_i+\right.\\\left.
    +\sum_{i\in\Lambda_r}(\sqrt{h_r}z^r_i+h_r)\sigma_i\right)=\\=\sum_{r=1}^K\frac{\alpha_r}{N_r}\mathbb{E}\log\sum_{\sigma\in\Sigma_{N_r}}\exp\left(
    \sum_{i\in\Lambda_r}\left(J_i^r\sqrt{Q_{\epsilon,r}(1)+h_r}+Q_{\epsilon,r}(1)+h_r\right)\sigma_i\right)
\end{multline*} where the last equality follows from the fact that $J_i^r$ and $z^r_i$ are independent standard Gaussian random variables. Finally:
\begin{align*}
    \bar{p}_{N,\epsilon}(1)=\sum_{r=1}^K\alpha_r\mathbb{E}_z\log2\cosh\left[
    z\sqrt{Q_{\epsilon,r}(1)+h_r}+Q_{\epsilon,r}(1)+h_r
    \right],\quad z\sim\mathcal{N}(0,1)\,.
\end{align*} By \eqref{corr_in_1} the derivatives of the pressure w.r.t. magnetic fields are bounded by $\alpha_r$ and then we can get rid of the explicit dependence on $\epsilon_r$ at the expense of a term $\mathcal{O}(s_N)$, thus getting \eqref{extreme_1_press}.

Analogously, by setting $t=0$, the interpolating Hamiltonian simply reduces to the original one \eqref{H_MSK_NL-bis} except for the $\epsilon_r$'s that can be neglected again at the expense of terms $\mathcal{O}(s_N)$.
\end{proof}

\begin{proposition}[Sum rule]
For any choice of the function $\mathbf{q}_\epsilon(t)$, the quenched pressure of the model \eqref{Q_press} obeys to the following sum rule:
\begin{multline}
    \label{Sum_rule_MSK_NL}
    \bar{p}_N(\mu,h)=\mathcal{O}(s_N)+\sum_{r=1}^K\alpha_r\psi(Q_{\epsilon,r}(1)+h_r)+\\
    +\int_0^1dt\,\left[\frac{
    (\mathbf{1}-\mathbf{q}_\epsilon(t),\Delta(\mathbf{1}-\mathbf{q}_\epsilon(t)))}{4}- \frac{(\mathbf{q}_\epsilon(t),\Delta\mathbf{q}_\epsilon(t))}{2}\right]+\frac{1}{4}\int_0^1dt\,R_\epsilon(t,\mu,h)
\end{multline}
where the remainder is:
\begin{align}\label{remainder_sum_rulw}
    R_\epsilon(t,\mu,h)=\mathbb{E}\Big\langle
    (\mathbf{m}-\mathbf{q}_\epsilon(t),\Delta(\mathbf{m}-\mathbf{q}_\epsilon(t)))\Big\rangle_{N,t}^{(\epsilon)}\,.
\end{align}
\end{proposition}

\begin{proof}
The proof consists in computing the first derivative by using Gaussian integration by parts for the terms containing the disorder.
\begin{multline*}
    \dot{\bar{p}}_{N,\epsilon}(t)=-\frac{1}{4}\mathbb{E}\Big\langle
    (\mathbf{1},\Delta\mathbf{1})-(\mathbf{q},\Delta\mathbf{q})
    \Big\rangle_{N,t}^{(\epsilon)}-\frac{1}{2}\mathbb{E}\Big\langle
    (\mathbf{m}-\mathbf{q}_\epsilon(t),\Delta(\mathbf{m}-\mathbf{q}_\epsilon(t)))
    \Big\rangle_{N,t}^{(\epsilon)}+\\+
    \frac{1}{2}(\mathbf{q}_\epsilon(t),\Delta\mathbf{q}_\epsilon(t))+\frac{1}{2}\mathbb{E}\Big\langle
    (\mathbf{1},\Delta\mathbf{q}_\epsilon(t))-
    (\mathbf{q}_\epsilon(t),\Delta\mathbf{q})
    \Big\rangle_{N,t}^{(\epsilon)}=\\=
    -\frac{1}{4}
    (\mathbf{1}-\mathbf{q}_\epsilon(t),\Delta(\mathbf{1}-\mathbf{q}_\epsilon(t)))+\frac{1}{2}
    (\mathbf{q}_\epsilon(t),\Delta\mathbf{q}_\epsilon(t)) +\\+\frac{1}{4}\mathbb{E}\Big\langle
    (\mathbf{q}-\mathbf{q}_\epsilon(t),\Delta(\mathbf{q}-\mathbf{q}_\epsilon(t)))
    \Big\rangle_{N,t}^{(\epsilon)}-\frac{1}{2}\mathbb{E}\Big\langle
    (\mathbf{m}-\mathbf{q}_\epsilon(t),\Delta(\mathbf{m}-\mathbf{q}_\epsilon(t)))
    \Big\rangle_{N,t}^{(\epsilon)}
\end{multline*}Using the Nishimori identities \eqref{N_identity_1} and \eqref{N_identity_2} we can sum the last two terms together:
\begin{multline}
    \dot{\bar{p}}_{N,\epsilon}(t)= -\frac{1}{4}
    (\mathbf{1}-\mathbf{q}_\epsilon(t),\Delta(\mathbf{1}-\mathbf{q}_\epsilon(t))) +\frac{1}{2}
    (\mathbf{q}_\epsilon(t),\Delta\mathbf{q}_\epsilon(t))+\\ -\frac{1}{4}\underbrace{\mathbb{E}\Big\langle
    (\mathbf{m}-\mathbf{q}_\epsilon(t),\Delta(\mathbf{m}-\mathbf{q}_\epsilon(t)))
    \Big\rangle_{N,t}^{(\epsilon)}}_{R_\epsilon(t,\mu,h)}\,.
\end{multline}
The sum rule then follows from a simple application of the Fundamental Theorem of Calculus and the previous Lemma:
\begin{align}
    \bar{p}_{N,\epsilon}(0)=\mathcal{O}(s_N)+\bar{p}_N(\mu,h)=\bar{p}_{N,\epsilon}(1)-\int_0^1dt\,\dot{\bar{p}}_{N,\epsilon}(t)\,.
\end{align}

\end{proof}

\section{Solution of the model}\label{section_4}
In this section we present the main result of the paper, namely the thermodynamic limit of the model under the hypothesis of a positive semi-definite effective interaction matrix: $\Delta\geq 0$. First, we need a couple of lemmas listed below.

\begin{lemma}[Liouville's formula]\label{Liouville_formula}Consider two matrices whose elements depend on a real parameter: $\Phi(t),\,A(t)$. Suppose that $\Phi$ satisfies:
\begin{align}
    \label{Liouville_edo}
    &\dot{\Phi}(t)=A(t)\Phi(t)\\
    &\Phi(0)=\Phi_0
\end{align}Then:
\begin{align}
    \text{\emph{det}}(\Phi(t))=\text{\emph{det}}(\Phi_0)\exp\left\{
    \int_0^tds\,\text{\emph{Tr}}(A(s))
    \right\}
\end{align}
\end{lemma}
\begin{definition}[Regularity of $\boldsymbol{\epsilon}\longmapsto\mathbf{Q}_{\epsilon}(\cdot)$]\label{def_regularity}
We will say that the map $\boldsymbol{\epsilon}\longmapsto\mathbf{Q}_{\epsilon}(\cdot)$ is \emph{regular} if
\begin{align}
    \text{det} \left(\frac{\partial\mathbf{Q_\epsilon}(t)}{\partial\boldsymbol{\epsilon}}\right)\geq 1\quad\forall t\in[0,1]
\end{align}
\end{definition}
\begin{remark}
Choosing $\mathbf{Q}_{\epsilon}$ as the solution of the following ODE:
\begin{align}
\dot{\mathbf{Q}}_\epsilon(t)=\hat{\alpha}^{-1}\Delta\,\mathbb{E}\langle\mathbf{m}\rangle_{N,t}^{(\epsilon)},\;\quad  \mathbf{Q}_\epsilon(0)=\boldsymbol{\epsilon}  
\end{align}the map $\epsilon\longmapsto\mathbf{Q}_{\epsilon}(\cdot)$ turns out to be regular. Indeed
\begin{align}
    \frac{d}{dt}\frac{\partial\mathbf{Q_\epsilon}(t)}{\partial\boldsymbol{\epsilon}}=\underbrace{\frac{\partial}{\partial\mathbf{Q_\epsilon}(t)}\hat{\alpha}^{-1}\Delta\mathbb{E}\langle\mathbf{m}\rangle_{N,t}^{(\epsilon)}}_{=:A(t)} \frac{\partial\mathbf{Q_\epsilon}(t)}{\partial\boldsymbol{\epsilon}}\;;
\end{align}
since $Q_{\epsilon,r}(t)$ can be regarded as the variance of a magnetic field on the Nishimori line in \eqref{H_EMANUE} and the entries of $\hat{\alpha}^{-1}$ and $\Delta$ are non-negative we have:
\begin{align}
    \text{Tr}A(t)\geq 0\,,
\end{align}
by the correlation inequalities of type II \eqref{corr_in_2}, \eqref{mag_increasign_h}. Finally using Liouville's formula we get:
\begin{align}
    \det\left(\frac{\partial\mathbf{Q_\epsilon}(t)}{\partial\boldsymbol{\epsilon}}\right)=\text{det}\underbrace{\left(\frac{\partial\mathbf{Q_\epsilon}(0)}{\partial\boldsymbol{\epsilon}}\right)}_{=\mathbbm{1}}\exp\left\{\int_0^tds\,\text{\text{Tr}}(A(s))\right\}\geq 1
\end{align}
We stress that the sign of $\Delta$ plays no role yet, since we have used only the positivity of its entries so far.
\end{remark}

\begin{lemma}[Concentration]\label{concentration_lemma}
Suppose $\boldsymbol{\epsilon}\longmapsto\mathbf{Q}_\epsilon(\cdot)$ is a regular map. Consider the quantity:
\begin{align}
    \mathcal{L}_r:=\frac{1}{N_r}\sum_{i\in\Lambda_r}\left(\sigma_i+\frac{J_i^r\sigma_i}{2\sqrt{Q_{\epsilon,r}(t)}}\right),\;\quad J_i^r\iid\mathcal{N}(0,1)
\end{align}
and introduce the $\boldsymbol{\epsilon}$-average:
\begin{align}
    \mathbb{E}_{\boldsymbol{\epsilon}} [\cdot]=\prod_{r=1}^K\left(
    \frac{1}{s_N}\int_{s_N}^{2s_N}d\epsilon_r
    \right)(\cdot) \;.
\end{align}
We have:
\begin{align}
\mathbb{E}_{\boldsymbol{\epsilon}}\mathbb{E}\Big\langle\left(\mathcal{L}_r-\mathbb{E}\langle\mathcal{L}_r\rangle_{N,t}^{(\epsilon)}\right)^2\Big\rangle_{N,t}^{(\epsilon)}\longrightarrow 0\;,\quad\text{when }N\to\infty
\end{align}
and
\begin{align}
    \label{Concentration_1}
    \mathbb{E}\Big\langle\left(\mathcal{L}_r-\mathbb{E}\langle\mathcal{L}_r\rangle_{N,t}^{(\epsilon)}\right)^2\Big\rangle_{N,t}^{(\epsilon)} \,\geq\, \frac{1}{4}
    \mathbb{E}\Big\langle\left(m_r-\mathbb{E}\langle m_r\rangle_{N,t}^{(\epsilon)}\right)^2\Big\rangle_{N,t}^{(\epsilon)}
\end{align}
therefore the magnetization (or the overlap) concentrates in $\boldsymbol{\epsilon}$-average.
\end{lemma}

The proof, simple but lengthy (see the Appendix), is based on controlling the thermal and disorder-related fluctuations of $\mathcal{L}_r$. This implies the control of the fluctuations of the magnetization thus ensuring the replica symmetry of the model, which is independent of positive definiteness of $\Delta$ and depends only on the positivity of its elements.

We have laid the ground for our main result: the computation of the quenched pressure in the thermodynamic limit in form of a finite dimensional (due to the concentration lemma) variational principle.
\begin{theorem}[Thermodynamic limit]\label{main_theorem} On the Nishimori line, when $\Delta\geq 0$, the thermodynamic limit of the pressure $\bar{p}(\mu,h):=\lim_{N\to\infty}\bar{p}_N(\mu,h)$ exists and:
\begin{align}\label{solution_of_the_model}
    \bar{p}(\mu,h)=\sup_{\mathbf{x}\in\mathbb{R}_{\geq 0}^K}\bar{p}(\mu,h;\mathbf{x})
\end{align}where
\begin{align}\label{p_variaz_MSKNL}
    \bar{p}(\mu,h;\mathbf{x}):=
    \frac{(\mathbf{1}-\mathbf{x},\Delta(\mathbf{1}-\mathbf{x}))}{4}-\frac{(\mathbf{x},\Delta\mathbf{x})}{2}
    + \sum_{r=1}^K\alpha_r\psi((\hat{\alpha}^{-1}\Delta\mathbf{x}+\mathbf{h})_r)
\end{align}with the following stationary condition:
\begin{align}
    \label{consist_eq}
    \mathbf{x}-\mathbb{E}_z\tanh\left(z\sqrt{\hat{\alpha}^{-1}\Delta\mathbf{x}+\mathbf{h}}+\hat{\alpha}^{-1}\Delta\mathbf{x}+\mathbf{h}\right)\in\text{Ker}\Delta\,,\quad z\sim\mathcal{N}(0,1)
\end{align}
\end{theorem}

\begin{proof}
Let us divide the proof in two steps.
\paragraph{Lower Bound:} We initially fix $\mathbf{q}_\epsilon(t)=\mathbf{x}\in\mathbb{R}_{\geq 0}^K$ in \eqref{Sum_rule_MSK_NL}. Up to orders $\mathcal{O}(s_N)$ we find:
\begin{multline}\label{lower_bound_proof}
    \bar{p}_N(\mu,h)=\mathcal{O}(s_N)+\frac{(\mathbf{1}-\mathbf{x},\Delta(\mathbf{1}-\mathbf{x}))}{4}-\frac{(\mathbf{x},\Delta\mathbf{x})}{2}
    +\\+ \sum_{r=1}^K\alpha_r\mathbb{E}_z\log2\cosh\left(z\sqrt{(\hat{\alpha}^{-1}\Delta\mathbf{x}+\mathbf{h})_r}+(\hat{\alpha}^{-1}\Delta\mathbf{x}+\mathbf{h})_r\right)+\\+\frac{1}{4}\int_0^1dt\,R_\epsilon(t,\mu,h)
\end{multline}We have exploited the result in Lemma \ref{pressione_estremi}.
Being $\Delta$ positive semi-definite, the rest has a positive sign, for it is a quadratic form exactly with matrix $\Delta$.

Hence:
\begin{align*}
    \bar{p}_N(\mu,h)\geq \mathcal{O}(s_N)+\frac{(\mathbf{1}-\mathbf{x},\Delta(\mathbf{1}-\mathbf{x}))}{4}-\frac{(\mathbf{x},\Delta\mathbf{x})}{2}
    +\sum_{r=1}^K\alpha_r\psi((\hat{\alpha}^{-1}\Delta\mathbf{x}+\mathbf{h})_r)
\end{align*}Then, taking the $\liminf_{N\to\infty}$ on both sides and optimizing with $\sup_\mathbf{x}$ we get the first bound:
\begin{align}
    \liminf_{N\to\infty}\bar{p}_N\geq\sup_\mathbf{x}\left\{
    \frac{(\mathbf{1}-\mathbf{x},\Delta(\mathbf{1}-\mathbf{x}))}{4}-\frac{(\mathbf{x},\Delta\mathbf{x})}{2}
    +\sum_{r=1}^K\alpha_r\psi((\hat{\alpha}^{-1}\Delta\mathbf{x}+\mathbf{h})_r)
    \right\}\,.
\end{align}

\paragraph{Upper Bound:} We  start with a key observation: $\psi(\cdot)$ is a convex function. It can be seen as a consequence of the correlation inequalities \eqref{corr_in_1} \eqref{corr_in_2} on the Nishimori line. In fact, $\psi(Q)$ can be recast in the following way:
\begin{align*}
    \psi(Q)=\mathbb{E}_z\log\sum_{\sigma=\pm1}e^{\sigma(z\sqrt{Q}+Q)}=\mathbb{E}_{z(Q)}\log\sum_{\sigma=\pm1}e^{\sigma z(Q)},\;z(Q)\sim\mathcal{N}(Q,Q)\,.
\end{align*}This is a simple 1-particle, free system on the Nishimori line. For this model we have:
\begin{align}\label{properties_psi}
    \frac{\partial \psi}{\partial Q}=\frac{1}{2}\mathbb{E}_z[1+\langle\sigma\rangle]\,,\quad\frac{\partial^2 \psi}{\partial Q^2}=\frac{1}{2}\mathbb{E}[(1-\langle\sigma\rangle^2)^2]\,,\quad\langle\sigma\rangle=\frac{\sum_{\sigma=\pm1}e^{\sigma z(Q)}\sigma}{\sum_{\sigma=\pm1}e^{\sigma z(Q)}}\,.
\end{align} This allows us to use Jensen's inequality to extract the integral in $Q_{\epsilon,r}(1)$ from the terms containing $\psi$ in \eqref{extreme_1_press} (in Lemma \ref{pressione_estremi}). More explicitly
\begin{align*}
    \sum_{r=1}^K\alpha_r\psi\left(\left(\hat{\alpha}^{-1}\Delta \int_0^1\mathbf{q}_\epsilon(t)\,dt+\mathbf{h}\right)_{\!r}\right)\leq\sum_{r=1}^K\alpha_r\int_0^1\psi\left(\left(\hat{\alpha}^{-1}\Delta \mathbf{q}_\epsilon(t)+\mathbf{h}\right)_{\!r}\right)\,dt\,.
\end{align*}
By inserting the previous inequality in the sum rule \eqref{Sum_rule_MSK_NL} we have that:
\begin{multline}
    \bar{p}_N(\mu,h)\leq \mathcal{O}(s_N)+\int_0^1dt\,  \left[\frac{(\mathbf{1}-\mathbf{q}_\epsilon(t),\Delta(\mathbf{1}-\mathbf{q}_\epsilon(t)))}{4}-\frac{(\mathbf{q}_\epsilon(t),\Delta\mathbf{q}_\epsilon(t))}{2}\right.
    +\\+ \left.\sum_{r=1}^K\alpha_r\psi((\hat{\alpha}^{-1}\Delta\mathbf{q}_\epsilon(t)+\mathbf{h})_r)\right]+\frac{1}{4}\int_0^1dt\,R_\epsilon(t,\mu,h)\leq\\\leq
    \mathcal{O}(s_N)+\sup_{\mathbf{x}}\left\{
    \frac{(\mathbf{1}-\mathbf{x},\Delta(\mathbf{1}-\mathbf{x}))}{4}-\frac{(\mathbf{x},\Delta\mathbf{x})}{2}
    + \sum_{r=1}^K\alpha_r\psi((\hat{\alpha}^{-1}\Delta\mathbf{x}+\mathbf{h})_r)
    \right\}+\\+\frac{1}{4}\int_0^1dt\,R_\epsilon(t,\mu,h)
\end{multline}
If we finally take the expectation $\mathbb{E}_{\boldsymbol{\epsilon}}$ on both sides of the previous inequality we get:
\begin{multline}\label{aux_1}
    \bar{p}_N(\mu,h)\leq\mathcal{O}(s_N)+\sup_{\mathbf{x}}\left\{
    \frac{(\mathbf{1}-\mathbf{x},\Delta(\mathbf{1}-\mathbf{x}))}{2}-(\mathbf{x},\Delta\mathbf{x})
    +\right.\\\left.+ \sum_{r=1}^K\alpha_r\psi((\hat{\alpha}^{-1}\Delta\mathbf{x}+\mathbf{h})_r)
    \right\}+\frac{1}{2}\mathbb{E}_{\boldsymbol{\epsilon}}\int_0^1dt\,R_\epsilon(t,\mu,h)
\end{multline}
Recall that the remainder, defined in \eqref{remainder_sum_rulw}, depends on the functions $\mathbf{q}_\epsilon(t)$. This time we choose $\mathbf{q}_\epsilon(t)$ according to a different criterion. We would like to have: $\Delta\mathbf{q}_\epsilon(t)=\Delta\mathbb{E}\langle\mathbf{m}\rangle^{(\epsilon)}_{N,t}$. In this way we could use the concentration Lemma \ref{concentration_lemma}. This can be achieved through the following ODE:
\begin{align}
    \dot{\mathbf{Q}}_\epsilon(t)=\hat{\alpha}^{-1}\Delta\mathbb{E}\langle\mathbf{m}\rangle_{N,t}^{(\epsilon)}=:\mathbf{F}(t,\mathbf{Q}_\epsilon(t)),\quad\mathbf{Q}_\epsilon(0)=\boldsymbol{\epsilon}
\end{align}
As seen in \eqref{corr_in_2}, the derivatives of $\mathbf{F}$ are positive and bounded for any fixed $N$. This guarantees the existence of a unique solution over $[0,1]$.

Then, exchanging the two integrals by Fubini's theorem in \eqref{aux_1}, and applying Lemma \ref{concentration_lemma} we get:
\begin{align*}
    \limsup_{N\to\infty}\bar{p}_N\leq \sup_{\mathbf{x}}\left\{
    \frac{(\mathbf{1}-\mathbf{x},\Delta(\mathbf{1}-\mathbf{x}))}{4}-\frac{(\mathbf{x},\Delta\mathbf{x})}{2}
    + \sum_{r=1}^K\alpha_r\psi((\hat{\alpha}^{-1}\Delta\mathbf{x}+\mathbf{h})_r)
    \right\}\,.
\end{align*}
The two bounds match and this proves \eqref{solution_of_the_model}. Moreover, using the properties \eqref{properties_psi} the gradient of \eqref{solution_of_the_model} is:
\begin{multline}\label{gradient_p_var}
    \nabla_{\mathbf{x}}\bar{p}(\mu,h;\mathbf{x})=-\frac{\Delta}{2}(\mathbf{1}-\mathbf{x})-\Delta\mathbf{x}+\frac{\Delta}{2}\mathbf{1}+\\+\frac{\Delta}{2}\mathbb{E}_z\tanh\left(z\sqrt{\hat{\alpha}^{-1}\Delta\mathbf{x}+\mathbf{h}}+\hat{\alpha}^{-1}\Delta\mathbf{x}+\mathbf{h}\right)=\\=
    \frac{\Delta}{2}\left[-\mathbf{x}+\mathbb{E}_z\tanh\left(z\sqrt{\hat{\alpha}^{-1}\Delta\mathbf{x}+\mathbf{h}}+\hat{\alpha}^{-1}\Delta\mathbf{x}+\mathbf{h}\right)\right]
\end{multline}and it vanishes exactly when \eqref{consist_eq} holds.
\end{proof}

\begin{remark}
Notice that the positive definiteness of $\Delta$ is used to ensure the positivity of the remainder in \eqref{lower_bound_proof}, that ultimately leads to the lower bound. It is evident that if the sign of $\Delta$ is not definite the technique used does not produce any bound. In that case indeed, there is a direction along which the quadratic form in \eqref{p_variaz_MSKNL} can blow up to infinity. Thus, one should expect a $\min-\max$ principle as happens for bipartite systems, \emph{e.g.} the Wishart model \cite{wigner-wishart} and the bipartite SK in its replica symmetric phase \cite{Alberici2020AnnealingAR,noi_deep2,bgg,genovese}, that are a paradigm for non-elliptic interaction structure.
\end{remark}

\begin{proposition}\label{prop_eq_consist}Let $\Delta$ be strictly positive definite in \eqref{p_variaz_MSKNL}. Denote by $\rho(A)$ the spectral radius of a matrix $A$ and by $\mathscr{H}_\mathbf{x}\bar{p}$ the Hessian matrix of $\bar{p}$. The following implication holds:
\begin{align}\label{concavity_pvar_cond}
    \rho(\hat{\alpha}^{-1}\Delta)<1\quad\Rightarrow\quad
    \mathscr{H}_\mathbf{x}\bar{p}(\mu,h;\mathbf{x})<0,\;\forall\mathbf{x}\in\mathbb{R}^K_{\geq 0}
\end{align}or equivalently $\bar{p}(\mu,h;\mathbf{x})$ is strictly concave w.r.t. $\mathbf{x}$.
\end{proposition}
\begin{proof}
The Hessian matrix can be computed starting from the gradient \eqref{gradient_p_var} and using properties \eqref{properties_psi}:
\begin{multline}\label{hessian_p_var}
    \mathscr{H}_\mathbf{x}\bar{p}(\mu,h;\mathbf{x})=-\frac{\Delta}{2}+\frac{1}{2}\Delta\mathcal{D}(\mathbf{x},\mathbf{h})\hat{\alpha}^{-1}\Delta=\\=\frac{1}{2}\Delta^{1/2} \left[-\mathbbm{1}+\Delta^{1/2}\hat{\alpha}^{-1}\mathcal{D}(\mathbf{x},\mathbf{h})\Delta^{1/2}\right]\Delta^{1/2}
\end{multline}
\begin{align}
    \mathcal{D}(\mathbf{x},\mathbf{h}):=\text{diag}\left\{\mathbb{E}_z
    \left(1-\tanh^2\left(z\sqrt{\hat{\alpha}^{-1}\Delta\mathbf{x}+\mathbf{h}}+\hat{\alpha}^{-1}\Delta\mathbf{x}+\mathbf{h}\right)_r\right)^2
    \right\}_{r=1,\dots,K}.
\end{align}Since similar matrices have the same spectral radius we have:
\begin{multline}
    \rho\left(\Delta^{1/2}\hat{\alpha}^{-1}\mathcal{D}(\mathbf{x},\mathbf{h})\Delta^{1/2}\right)=\rho(\hat{\alpha}^{-1}\mathcal{D}(\mathbf{x},\mathbf{h})\Delta)=\\=\rho\left(
    \mathcal{D}^{1/2}(\mathbf{x},\mathbf{h})\hat{\alpha}^{-1/2}\Delta\hat{\alpha}^{-1/2}\mathcal{D}^{1/2}(\mathbf{x},\mathbf{h})
    \right)\,.
\end{multline}Now we use the fact that, for symmetric matrices, the spectral radius coincides with the matrix norm induced by the Euclidean norm:
\begin{multline}
    \rho\left(\Delta^{1/2}\hat{\alpha}^{-1}\mathcal{D}(\mathbf{x},\mathbf{h})\Delta^{1/2}\right)=\Vert\mathcal{D}^{1/2}(\mathbf{x},\mathbf{h})\hat{\alpha}^{-1/2}\Delta\hat{\alpha}^{-1/2}\mathcal{D}^{1/2}(\mathbf{x},\mathbf{h})\Vert\leq \\\leq\Vert\mathcal{D}(\mathbf{x},\mathbf{h})\Vert\Vert\hat{\alpha}^{-1/2}\Delta\hat{\alpha}^{-1/2}\Vert=\Vert\mathcal{D}(\mathbf{x},\mathbf{h})\Vert\rho\left(\hat{\alpha}^{-1/2}\Delta\hat{\alpha}^{-1/2}\right)\leq\\\leq \rho\left(\hat{\alpha}^{-1/2}\Delta\hat{\alpha}^{-1/2}\right)\,.
\end{multline}Finally, exploiting again matrix similarity:
\begin{align}
    \rho\left(\Delta^{1/2}\hat{\alpha}^{-1}\mathcal{D}(\mathbf{x},\mathbf{h})\Delta^{1/2}\right)\leq \rho\left(\hat{\alpha}^{-1}\Delta\right)<1
\end{align}by hypothesis. The previous one implies that:
\begin{align}
    -\mathbbm{1}+\Delta^{1/2}\hat{\alpha}^{-1}\mathcal{D}(\mathbf{x},\mathbf{h})\Delta^{1/2}<0
\end{align}whence, for any test vector $\mathbf{v}$:
\begin{multline}
    \left(\mathbf{v},\Delta^{1/2} \left[-\mathbbm{1}+\Delta^{1/2}\hat{\alpha}^{-1}\mathcal{D}(\mathbf{x},\mathbf{h})\Delta^{1/2}\right]\Delta^{1/2}\mathbf{v}\right)=\\=\left(\Delta^{1/2}\mathbf{v}, \left[-\mathbbm{1}+\Delta^{1/2}\hat{\alpha}^{-1}\mathcal{D}(\mathbf{x},\mathbf{h})\Delta^{1/2}\right](\Delta^{1/2}\mathbf{v})\right)<0\;.
\end{multline}
\end{proof}
\begin{remark}
The previous proposition implies that, whenever $\Delta$ is invertible, $\mathbf{h}=0$ and $\rho(\hat{\alpha}^{-1}\Delta)<1$, the point $\mathbf{x}=0$ is the unique maximizer of \eqref{p_variaz_MSKNL}. On the contrary, when $\rho(\hat{\alpha}^{-1}\Delta)>1$ we have
\begin{align*}
    \mathscr{H}_\mathbf{x}\bar{p}(\mu,0;0)=\frac{1}{2}\Delta^{1/2} \left[-\mathbbm{1}+\Delta^{1/2}\hat{\alpha}^{-1}\Delta^{1/2}\right]\Delta^{1/2}
\end{align*}
and the matrix in square brackets has at least one positive eigenvalue, therefore $\mathbf{x}=0$ becomes an unstable saddle point for the variational pressure, thus signalling a phase transition. Notice that this instability can be generated both varying the parameters $\Delta_{rs}$ and the form factors $\alpha_r$.
\end{remark}
\begin{remark}
If $\Delta$ is non singular, our variational pressure \eqref{p_variaz_MSKNL} goes to $-\infty$ as $\Vert\mathbf{x}\Vert\to\infty$, because the concave quadratic form always dominates the sum of the terms containing $\psi$, which is Lispchitz with $\text{Lip}(\psi)\leq 1$ (again by \eqref{properties_psi}). This, together with the regularity of $\bar{p}$ ensures that there is a global maximum satisfying the fixed point equation:
\begin{align}\label{fixed_point_eq_h00}
    \mathbf{x}=\mathbb{E}_z\tanh\left(z\sqrt{\hat{\alpha}^{-1}\Delta\mathbf{x}+\mathbf{h}}+\hat{\alpha}^{-1}\Delta\mathbf{x}+\mathbf{h}\right)=:\mathbf{T}(\mathbf{x};\mathbf{h})\;.
\end{align}
The Jacobian matrix of $\mathbf{T}(\cdot;\mathbf{h})$ is:
\begin{align}
    \text{D}\mathbf{T}(\mathbf{x};\mathbf{h})=\mathcal{D}(\mathbf{x},\mathbf{h})\hat{\alpha}^{-1}\Delta
\end{align}
and satisfies:
\begin{align}
    \rho(\text{D}\mathbf{T}(\mathbf{x};\mathbf{h}))=\rho(\mathcal{D}(\mathbf{x},\mathbf{h})\hat{\alpha}^{-1}\Delta)\leq \rho(\hat{\alpha}^{-1}\Delta)
\end{align}as proved in Proposition \ref{prop_eq_consist}. Equality holds at $\mathbf{h}=0$ and $\mathbf{x}=0$. Hence when $\rho(\hat{\alpha}^{-1}\Delta)<1$ the iteration of $\mathbf{T}(\cdot;\mathbf{h})$ converges to a fixed point. If this does not hold, we still have that at one local maximum point, say $\mathbf{x}^*$:
\begin{align}
    \mathscr{H}_\mathbf{x}\bar{p}(\mu,h;\mathbf{x}^*)<0\quad\text{or}\quad\rho(\Delta^{1/2}\hat{\alpha}^{-1}\mathcal{D}(\mathbf{x}^*,\mathbf{h})\Delta^{1/2})=\rho(\text{D}\mathbf{T}(\mathbf{x}^*;\mathbf{h}))<1\;.
\end{align}The latter implies that the iteration $\mathbf{x}_{n+1}=\mathbf{T}(\mathbf{x}_n;\mathbf{h})$ converges to $\mathbf{x}^*$ (locally) provided that $\Vert\mathbf{x}_0-\mathbf{x}^*\Vert<\delta$ with $\delta$ sufficiently small. 
\end{remark}
\begin{remark}
Our parameters lie in $\mathbb{R}_{\geq0}^K$, thus the vanishing gradient condition \emph{a priori} allows us only to find maximizers of \eqref{p_variaz_MSKNL} in the interior, namely when $x_r>0\,\forall\,r=1,\dots,K$. More rigorously, the necessary conditions for a point $\bar{\mathbf{x}}\in\mathbb{R}_{\geq0}^K$ to be a maximizer are:
\begin{align}
    \begin{cases}
    \partial_{x_r}\bar{p}(\mu,h;\bar{\mathbf{x}})=\frac{1}{2}\left[
    \Delta(-\bar{\mathbf{x}}+\mathbf{T}(\bar{\mathbf{x}};\mathbf{h}))
    \right]_r\leq 0\\
    \bar{x}_r\partial_{x_r}\bar{p}(\mu,h;\bar{\mathbf{x}})=\frac{1}{2}x_r\left[
    \Delta(-\bar{\mathbf{x}}+\mathbf{T}(\bar{\mathbf{x}};\mathbf{h}))
    \right]_r= 0
    \end{cases}
\end{align}If we notice that $T_r(\mathbf{x};\mathbf{h})\geq 0$ these conditions imply:
\begin{multline}
    \begin{cases}
    \left(\mathbf{T}(\bar{\mathbf{x}};\mathbf{h}),\Delta(-\bar{\mathbf{x}}+\mathbf{T}(\bar{\mathbf{x}};\mathbf{h}))\right)\leq0\\
    \left(\bar{\mathbf{x}},\Delta(-\bar{\mathbf{x}}+\mathbf{T}(\bar{\mathbf{x}};\mathbf{h}))\right)=0
    \end{cases}\quad\Rightarrow\\\Rightarrow\quad
    \left(-\bar{\mathbf{x}}+\mathbf{T}(\bar{\mathbf{x}};\mathbf{h}),\Delta(-\bar{\mathbf{x}}+\mathbf{T}(\bar{\mathbf{x}};\mathbf{h}))\right)\leq 0\;.
\end{multline}However, since $\Delta>0$ we must necessarily have:
\begin{align}
    \left(-\bar{\mathbf{x}}+\mathbf{T}(\bar{\mathbf{x}};\mathbf{h}),\Delta(-\bar{\mathbf{x}}+\mathbf{T}(\bar{\mathbf{x}};\mathbf{h}))\right)=0\quad\Leftrightarrow\quad-\bar{\mathbf{x}}+\mathbf{T}(\bar{\mathbf{x}};\mathbf{h})=0\;.
\end{align}From the previous we can see that the consistency equation \eqref{fixed_point_eq_h00} is necessarily satisfied also by maximizers on the boundary. 
\end{remark}

\section{Perspectives and conclusions}

It is interesting to emphasize the link of our model with an inference problem since, as we have seen, among the techniques we use 
there are some whose origins are within high dimensional statistical inference. This fact goes beyond a bare technical analogy. Indeed our model admits itself an inferential interpretation. To start with, it is well known \cite{adaptive} that  our case $K=1$, \emph{i.e.} SK on the Nishimori line which is also called planted SK model, is equivalent, thanks to the $\mathbb{Z}_2$ gauge symmetry, to the Wigner spiked model  with Rademacher prior $\rho= 1/2(\delta_{1}+\delta_{-1})$. 
For generic $K$ the corresponding inference problem is defined as follows. Given a family of non negative numbers $(\mu_{rs})_{r,s=1,\dots,K}$, consider a Gaussian channel
\begin{align}
    y_{ij}(\mu_{rs})=\sqrt{\frac{\mu_{rs}}{2N}}\sigma^*_i\sigma^*_j+z_{ij}\,,\quad z_{ij}\iid\mathcal{N}(0,1)\,,\,(i,j)\in\Lambda_r\times\Lambda_s
\end{align}
where $\sigma^*\in \{\pm1\}^N$ (the ground truth) is the signal we want to recover through the observations $y_{ij}$.  Here $\mu_{rs}$ play the role of an index dependent signal-to-noise ratio. 
Up to costants, the Gibbs measure associated to our Hamiltonian \eqref{H_MSK_NL} corresponds
to the posterior distribution in the Bayesian optimal setting
and the pressure corresponds to the mutual information. This correspondence can be obtained along the same lines of the case $K=1$ and we refer to \cite{adaptive} for the details.

The  model that we take into account was studied under some 
specific assumptions on the $\mu_{rs}$, listed in Paragraph 2.3 in \cite{Jean-lenka2} (see also \cite{Lenka}). 
The thermodynamic properties the authors focus on are obtained by first considering the infinite volume limit of each block and then sending the number of blocks to infinity thus recovering the limiting mutual information of the Wigner spiked model, i.e. the case with homogeneous ($\mu_{rs}=\bar{\mu}$) signal-to-noise ratio. In the present work instead, in the case of a Rademacher planted signal and the only positive definiteness assumption on the matrix $\mu$, the model is studied and solved for arbitrary number of species and form factors, through a replica symmetry result and a finite dimensional variational principle for the model pressure in the infinite volume limit. 
The positivity assumption on $\mu$ rules out some interesting non-elliptic structures such as restricted Boltzmann machines. However, we will show in a follow up work how to deal with these non convexities also proving a replica symmetric variational formula for the pressure of the Deep Boltzmann Machine on the Nishimori line \cite{DBMNL}.\\  
\vspace{10pt}

\textbf{Acknowledgements} 
The authors thank Jean Barbier, Adriano Barra, Francesco Guerra and Nicolas Macris for interesting discussions. In particular we thank Jean Barbier for pointing out the inferential interpretation of our model. D.A. and E.M. acknowledge support from Progetto Alma Idea 2018, Università di Bologna.

\appendix
\section{Appendix: Proof of the Concentration Lemma}\label{Appendix_A}
\numberwithin{equation}{section}
\setcounter{equation}{0}

\begin{proof}[Proof of Lemma \ref{concentration_lemma}]
Let us split the proof into three steps for the sake of clarity. As anticipated, it is convenient to split the total fluctuation of $\mathcal{L}_r$ into two parts, thus proving that:
\begin{align}
    \label{Concentration_2}
    &\mathbb{E}_{\boldsymbol{\epsilon}}\mathbb{E}
    \Big\langle\left(\mathcal{L}_r-\langle\mathcal{L}_r\rangle_{N,t}^{(\epsilon)}\right)^2\Big\rangle_{N,t}^{(\epsilon)}\longrightarrow0\quad\text{as }N\to\infty\\
    \label{Concentration_3}
    &\mathbb{E}_{\boldsymbol{\epsilon}}\mathbb{E}
    \left(\langle\mathcal{L}_r\rangle_{N,t}^{(\epsilon)}-\mathbb{E}\langle\mathcal{L}_r\rangle_{N,t}^{(\epsilon)}\right)^2\longrightarrow0\quad\text{as }N\to\infty\;.
\end{align}
From this moment on, we neglect sub and superscripts in the Gibbs brackets as well as $t$-dependencies. We start by proving the inequality \eqref{Concentration_1}.

\paragraph{Proof of inequality \eqref{Concentration_1}:} To begin with, we compute:
\begin{align}
    &\mathbb{E}[\langle\mathcal{L}_r\rangle]=\frac{1}{N_r}\sum_{i\in\Lambda_r}\mathbb{E}\Big\langle
    \sigma_i+\frac{J_i^r\sigma_i}{2\sqrt{Q_{\epsilon, r}}}
    \Big\rangle=\mathbb{E}\langle m_r\rangle+\frac{1}{2}\mathbb{E}[1-\langle m_r\rangle]=\frac{1}{2}\mathbb{E}[1+\langle m_r\rangle]
\end{align}where integration by parts has been used.

Then, we proceed with:
\begin{multline}
    \mathbb{E}\langle\mathcal{L}_r^2\rangle=\frac{1}{N_r^2}\sum_{i,j\in\Lambda_r}\mathbb{E}\Big\langle
    \sigma_i\sigma_j+\frac{J_i^r\sigma_i\sigma_j}{\sqrt{Q_{\epsilon,r}}}+\frac{J_i^rJ_j^r\sigma_i\sigma_j}{4Q_{\epsilon,r}}
    \Big\rangle=\underbrace{\mathbb{E}\langle m_r^2\rangle}_{R_1}+\underbrace{ \frac{1}{N_r^2}\sum_{i,j\in\Lambda_r}\mathbb{E}\Big\langle\frac{J_i^r\sigma_i\sigma_j}{\sqrt{Q_{\epsilon,r}}}\Big\rangle}_{R_2}+\\+
    \underbrace{ \frac{1}{N_r^2}\sum_{i,j\in\Lambda_r}\mathbb{E}\Big\langle\frac{J_i^rJ_j^r\sigma_i\sigma_j}{4Q_{\epsilon,r}}\Big\rangle}_{R_3}\,.
\end{multline}We treat the three terms $R_1$, $R_2$ and $R_3$ separately with repeated integrations by parts.
\begin{align}
    R_2=\frac{1}{N_r^2}\sum_{i,j\in\Lambda_r}\mathbb{E}
    \left[
    \langle \sigma_j\rangle-\langle\sigma_i\sigma_j\rangle\langle\sigma_i\rangle
    \right]=\mathbb{E}\langle m_r\rangle-\mathbb{E}\langle m_r\rangle^2
\end{align}where we have used the Nishimori identity \eqref{N_identity_3}.
\begin{multline}
    R_3=\frac{1}{4N_r^2}\sum_{i,j\in\Lambda_r}\mathbb{E}\left[ \frac{\delta_{ij}\overbrace{\langle\sigma_i\sigma_j\rangle}^{=1}}{Q_{\epsilon,r}} +J_j^r\frac{ (\langle\sigma_j\rangle-\langle\sigma_i\rangle\langle\sigma_i\sigma_j\rangle)}{\sqrt{Q_{\epsilon,r}}}\right]=
    \frac{1}{4N_rQ_{\epsilon,r}}+\\+\frac{1}{4N_r^2}\sum_{i,j\in\Lambda_r}\mathbb{E}\left[ 1-\langle\sigma_j\rangle^2-\langle\sigma_i\rangle(\langle\sigma_i\rangle-\langle\sigma_j\rangle\langle\sigma_i\sigma_j\rangle)-\langle\sigma_i\sigma_j\rangle(\langle\sigma_i\sigma_j\rangle-\langle\sigma_i\rangle\langle\sigma_j\rangle)
    \right]=\\=
    \frac{1}{4N_rQ_{\epsilon,r}}+\frac{1}{4}-\frac{1}{2}\mathbb{E}\langle m_r\rangle+\frac{1}{2}\mathbb{E}\langle m_r\rangle^2-\frac{1}{4}\mathbb{E}\langle m_r^2\rangle\,.
\end{multline}
Hence:
\begin{align}
    R_1+R_2+R_3=\frac{1}{4}+\frac{1}{4N_rQ_{\epsilon,r}}+\frac{3}{4}\mathbb{E}\langle m_r^2\rangle+\frac{1}{2}\mathbb{E}\langle m_r\rangle-\frac{1}{2}\mathbb{E}\langle m_r\rangle^2\,.
\end{align}Summing up all the contributions:
\begin{multline}
    \mathbb{E}\langle\mathcal{L}_r^2\rangle-(\mathbb{E}\langle\mathcal{L}_r\rangle)^2=\frac{1}{4N_rQ_{\epsilon,r}}+\frac{3}{4}\mathbb{E}\langle m_r^2\rangle-\frac{1}{2}\mathbb{E}\langle m_r\rangle^2-\frac{1}{4}(\mathbb{E}\langle m_r\rangle)^2=\\=\frac{1}{4N_rQ_{\epsilon,r}}+\frac{1}{4}(\mathbb{E}\langle m_r^2\rangle-(\mathbb{E}\langle m_r\rangle)^2)+\frac{1}{2}(\mathbb{E}\langle m_r^2\rangle-\mathbb{E}\langle m_r\rangle^2)\geq\frac{1}{4}\mathbb{E}\Big\langle(m_r-\mathbb{E}\langle m_r\rangle)^2\Big\rangle\,.
\end{multline}

\paragraph{Proof of \eqref{Concentration_2}:} Notice that:
\begin{align}
    &\frac{\partial \bar{p}_{N,\epsilon}}{\partial Q_{\epsilon,r}}=\frac{1}{N}\mathbb{E}\Big\langle
    \sum_{i\in\Lambda_r}\left(
    \sigma_i+\frac{J_i^r\sigma_i}{2\sqrt{Q_{\epsilon,r}}}
    \right)
    \Big\rangle=\alpha_r\mathbb{E}\langle\mathcal{L}_r\rangle=\frac{\alpha_r}{2}\mathbb{E}[1+\langle m_r\rangle],\;\quad J_i^r\iid\mathcal{N}(0,1)\\
    &\frac{\partial^2 \bar{p}_{N,\epsilon}}{\partial Q_{\epsilon,r}^2}=\alpha_r N_r\mathbb{E}\langle(\mathcal{L}_r-\langle\mathcal{L}_r\rangle)^2\rangle-\frac{1}{4NQ_{\epsilon,r}^{3/2}}\sum_{i\in\Lambda_r}\mathbb{E}\langle J_i^r\sigma_i\rangle\,.
\end{align}
From the last one, after an integration by parts and using the regularity of the map $\epsilon\longmapsto\mathbf{Q}_\epsilon(\cdot)$ and Lemma \ref{Liouville_formula} we get:
\begin{multline}
    \mathbb{E}_{\boldsymbol{\epsilon}}\mathbb{E}\langle(\mathcal{L}_r-\langle\mathcal{L}_r\rangle)^2\rangle\leq \frac{1}{N_r\alpha_rs_N^K}\prod_{s=1}^K\int_{Q_{s_N,s}}^{Q_{2s_N,s}}dQ_{\epsilon,s}\frac{\partial^2 \bar{p}_{N,\epsilon}}{\partial Q_{\epsilon,r}^2}+\mathbb{E}_{\boldsymbol{\epsilon}}\frac{1}{4N_r\epsilon_r}\mathbb{E}[1-\langle m_r\rangle]\leq\\\leq
    \frac{1}{N_r\alpha_rs_N^K}\prod_{s\neq r,1}^K\int_{Q_{s_N,s}}^{Q_{2s_N,s}}dQ_{\epsilon,s}\left[\left.\frac{\partial \bar{p}_{N,\epsilon}}{\partial Q_{\epsilon,r}}\right|_{Q_{2s_N,r}}-\left.\frac{\partial \bar{p}_{N,\epsilon}}{\partial Q_{\epsilon,r}}\right|_{Q_{s_N,r}}\right]+\frac{\log2}{4N_rs_N}\leq\\\leq
    \frac{2K_r(\Delta)}{N_rs_N^K}+\frac{\log2}{4N_rs_N}=\mathcal{O}\left(
    \frac{1}{N_rs_N^K}
    \right)\longrightarrow0
\end{multline}where:
\begin{align}
    \prod_{s\neq r,1}^K(Q_{2s_N,s}-Q_{s_N,s})\leq K_r(\Delta)\,.
\end{align}

\paragraph{Proof of \eqref{Concentration_3}:}
Let $p_{N,\epsilon}$ be the random interpolating pressure, such that $\mathbb{E}p_{N,\epsilon}=\bar{p}_{N,\epsilon}$. Define:
\begin{align}
    &\hat{p}_{N,\epsilon}=p_{N,\epsilon}-\alpha_r\sqrt{Q_{\epsilon,r}}\sum_{i\in\Lambda_r}\frac{|J_i^r|}{N_r},\;\quad\hat{\bar{p}}_{N,\epsilon}=\mathbb{E}\hat{p}_{N,\epsilon}\\
    &\frac{\partial^2 \hat{p}_{N,\epsilon}}{\partial Q_{\epsilon,r}^2}=\alpha_r N_r\langle(\mathcal{L}_r-\langle\mathcal{L}_r\rangle)^2\rangle+\frac{\alpha_r}{4Q_{\epsilon,r}^{3/2}}\sum_{i\in\Lambda_r}\frac{|J_i^r|- J_i^r\langle\sigma_i\rangle}{N_r} \geq 0\,.
\end{align}
Let us evaluate:
\begin{align}\label{previous_1}
    \left|
    \frac{\partial \hat{p}_{N,\epsilon}}{\partial Q_{\epsilon,r}}-\frac{\partial \hat{\bar{p}}_{N,\epsilon}}{\partial Q_{\epsilon,r}}
    \right|\geq \alpha_r\left|\langle\mathcal{L}_r\rangle-\mathbb{E}\langle\mathcal{L}_r\rangle\right|-\frac{\alpha_r|A_r|}{2\sqrt{Q_{\epsilon,r}}}
\end{align}where:
\begin{align}\label{previous_2}
    A_r:=\frac{1}{N_r}\sum_{i\in\Lambda_r}[|J_i^r|-\mathbb{E}|J_i^r|]\,.
\end{align}
Thanks to the independence of the $J_i^r$ it is immediate to verify  that $\exists\, a\geq 0$ s.t.:
\begin{align}
    \mathbb{E}[A_r^2]\leq \frac{a}{N_r}\,.
\end{align}
Using Lemma 3.2 in \cite{panchenko2015sherrington}, with notations used in \cite{adaptive}:
\begin{multline}\label{previous_3}
    \left|
    \frac{\partial \hat{p}_{N,\epsilon}}{\partial Q_{\epsilon,r}}-\frac{\partial \hat{\bar{p}}_{N,\epsilon}}{\partial Q_{\epsilon,r}}
    \right|\leq\frac{1}{\delta}\sum_{u\in\{Q_{\epsilon,r}+\delta,Q_{\epsilon,r},Q_{\epsilon,r}-\delta\}}[|\hat{p}_{N,\epsilon}-\hat{\bar{p}}_{N,\epsilon}|+\alpha_r\sqrt{u}|A_r|]+\\+C_\delta^+(Q_{\epsilon,r})+C_\delta^-(Q_{\epsilon,r})
\end{multline}with:
\begin{align}
    &C_\delta^\pm(Q_{\epsilon,r})=|\hat{\bar{p}}'_{N,\epsilon}(Q_{\epsilon,r}\pm\delta)-\hat{\bar{p}}'_{N,\epsilon}(Q_{\epsilon,r})|\\
    &|\hat{\bar{p}}'_{N,\epsilon}|=\left|\frac{\alpha_r}{2}\mathbb{E}[1+\langle m_r\rangle]-\frac{\alpha_r\mathbb{E}|J_1^r|}{2\sqrt{Q_{\epsilon,r}}}\right|\leq \alpha_r\left(1+\frac{C}{2\sqrt{s_N}}\right)\\
    &C_\delta^\pm(Q_{\epsilon,r})\leq \alpha_r\left(2+\frac{C}{\sqrt{s_N}}\right)
\end{align}
where for simplicity we have kept the dependence on $Q_{\epsilon,r}$ only, $\hat{\bar{p}}'_{N,\epsilon}$ is the derivative w.r.t. it and $\delta>0$. Notice that $\delta$ will be chosen strictly smaller than $s_N$, so that $Q_{\epsilon,r}-\delta\geq \epsilon-\delta\geq s_N-\delta>0$.  

Then, using \eqref{previous_1}, \eqref{previous_2} and \eqref{previous_3}, and thanks to the fact that $(\sum_{i=1}^p\nu_i)^2\leq p\sum_{i=1}^p\nu_i^2$, we get:
\begin{multline}\label{semi_final}
    \frac{\alpha^2_r}{9}\left|\langle\mathcal{L}_r\rangle-\mathbb{E}\langle\mathcal{L}_r\rangle\right|^2\leq \frac{1}{\delta^2}\sum_{u\in\{Q_{\epsilon,r}+\delta,Q_{\epsilon,r},Q_{\epsilon,r}-\delta\}}[|\hat{p}_{N,\epsilon}-\hat{\bar{p}}_{N,\epsilon}|^2+\alpha_r^2u|A_r|^2]+\\+C_\delta^+(Q_{\epsilon,r})^2+C_\delta^-(Q_{\epsilon,r})^2+\frac{\alpha_r^2A_r^2}{4\epsilon_r}\,.
\end{multline}
We first evaluate the two terms containing $C_\delta^\pm$:
\begin{multline}
    \mathbb{E}_{\boldsymbol{\epsilon}}[C_\delta^+(Q_{\epsilon,r})^2+C_\delta^-(Q_{\epsilon,r})^2]\leq
    2\alpha_r\left(
    2+\frac{C}{\sqrt{s_N}}
    \right)
    \mathbb{E}_{\boldsymbol{\epsilon}}[C_\delta^+(Q_{\epsilon,r})+C_\delta^-(Q_{\epsilon,r})]\leq \\
    \leq
    \frac{2\alpha_r}{s_N^K}\left(
    2+\frac{C}{\sqrt{s_N}}
    \right)\prod_{s=1}^K\int_{Q_{s_N,s}}^{Q_{2s_N,s}}dQ_{\epsilon,s}[\hat{\bar{p}}'_{N,\epsilon}(Q_{\epsilon,r}+\delta)-\hat{\bar{p}}'_{N,\epsilon}(Q_{\epsilon,r}-\delta)]=\\=
     \frac{2\alpha_r}{s_N^K}\left(
    2+\frac{C}{\sqrt{s_N}}
    \right)\prod_{s\neq r,1}^K \int_{Q_{s_N,s}}^{Q_{2s_N,s}}dQ_{\epsilon,s}[\hat{\bar{p}}_{N,\epsilon}(Q_{2s_N,r}+\delta)-\hat{\bar{p}}_{N,\epsilon}(Q_{2s_N,r}-\delta)+\\-\hat{\bar{p}}_{N,\epsilon}(Q_{s_N,r}+\delta)+\hat{\bar{p}}_{N,\epsilon}(Q_{s_N,r}-\delta)]\leq
    \frac{8\alpha_r^2 K_r(\Delta)}{s_N^K}\delta\left(
    2+\frac{C}{\sqrt{s_N}}
    \right)^2\,.
\end{multline}

Taking the expectation $\mathbb{E}_{\boldsymbol{\epsilon}}\mathbb{E}$ in \eqref{semi_final}, and defining $W_r$ s.t. $Q_{\epsilon,r}\leq W_r$, we get:
\begin{multline}\label{conc}
    \frac{\alpha^2_r}{9}\mathbb{E}_{\boldsymbol{\epsilon}}\mathbb{E}\left|\langle\mathcal{L}_r\rangle-\mathbb{E}\langle\mathcal{L}_r\rangle\right|^2\leq \frac{3}{\delta^2}\left[\frac{S}{N}+\frac{\alpha_r^2W_ra}{N_r}\right]+\\+\frac{8\alpha_r^2 K_r(\Delta)}{s_N^K}\delta\left(
    2+\frac{C}{\sqrt{s_N}}
    \right)^{2}+\frac{\alpha_r^2a\log2}{4N_rs_N}\,.
\end{multline}
We can make the r.h.s. vanish by choosing for example: $\delta=s_N^{2K/3}N^{-1/3}$. The choice $s_N\propto N^{-1/16K}$ makes the r.h.s. \eqref{conc} behave like $\mathcal{O}(N^{-1/4})$.
\end{proof}

\section{Appendix: the SK case}\label{Appendix_B}
In the case $K=1$ the equation \eqref{solution_of_the_model} reduces to:
\begin{align}
    \label{monospecies_LTD}
    \lim_{N\to\infty}\bar{p}_N(\mu,h)=\sup_{x\in\mathbb{R}_{\geq 0}}\left\{
    \mu\frac{(1-x)^2}{4}-\frac{\mu x^2}{2}+\psi(\mu x+h)
    \right\}
\end{align}while \eqref{consist_eq} simply becomes:
\begin{align}\label{consist_eq_mono}
    x=\mathbb{E}_z\tanh\left(z\sqrt{\mu x+h}+\mu x+h\right):=T(x;\mu,h)\;.
\end{align}

We collect the main results on this model in the following proposition.
\begin{proposition}
Define:
\begin{align}
    \bar{p}_{var}(x;\mu,h)=\mu\frac{(1-x)^2}{4}-\frac{\mu x^2}{2}+\psi(\mu x+h)\;.
\end{align}The following hold:
\begin{enumerate}
    \item if $\mu<1$ then $\bar{p}_{var}$ is concave in x. Equivalently if $\mu<1$ then $T(x;\mu,h)$ is a contraction, and if further $h=0$ then $x=0$ is its fixed point;
    \item the stable solution of the consistency equation \eqref{consist_eq_mono} is continuous at $(\mu,h)=(1,0)$:
    \begin{align}
        \lim_{(\mu,h)\to (1,0)}\bar{x}(\mu,h)=0=\bar{x}(1,0)\;;
    \end{align}
    \item for fixed $h=0$, the magnetization goes to $0$ linearly with $\mu-1$ as $\mu\to1_+$, more precisely:
    \begin{align}
        \bar{x}=(1+o(1))\frac{\mu-1}{\mu^2}
    \end{align}where $o(1)$ goes to $0$  when $\mu\to1_+$. Therefore the critical exponent $\beta$ (in the Landau classification) is $1$, which means that the derivative of the magnetization w.r.t. $\mu$ does not diverge at the critical point, it only jumps from $0$ to $1$ and then decreases;
    \item Along the line $(\mu,\lambda(\mu-1)),\;\lambda>0$ in the plane $(\mu,h)$ the magnetization goes to $0$ as follows:
    \begin{align}
        \bar{x}=\sqrt{\frac{\lambda(\mu-1)}{\mu^2}}(1+o(1))\;
    \end{align}when $\mu\to1_+$, therefore with a critical exponent $1/2$;
    \item For fixed $\mu=1$ and $h\to 0_+$ the magnetization behaves as:
    \begin{align}
        \bar{x}^2=h(1+o(1))
    \end{align}
    where $o(1)\to 0$ when $h\to0_+$. Therefore we have a critical exponent $\delta=2$ (according to Landau's classification). 
\end{enumerate}
\end{proposition}
\begin{proof}
\emph{1.}  The fist assertion follows immediately from \eqref{concavity_pvar_cond}, since $\hat{\alpha}\equiv 1$ and $\Delta\equiv\mu$. Then, by \eqref{properties_psi}:
\begin{align*}
    \frac{dT}{dx}(x;\mu,h)=\mu\mathbb{E}_z\left[\left(1-\tanh^2\left(z\sqrt{\mu x+h}+\mu x+h\right)\right)^2\right]\leq \mu<1
\end{align*}that implies $T$ is a contraction. It is easy to see that if $h=0$ then $x=0$ is a solution of the fixed point equation which must be unique by Banach's fixed point theorem.

\emph{2.} Using continuity and monotonicity of T (see \eqref{properties_psi}):
\begin{align*}
    &\limsup_{(\mu,h)\to (1,0)}\bar{x}(\mu,h)=T(\limsup_{(\mu,h)\to (1,0)}\bar{x}(\mu,h);1,0)\\
    &\liminf_{(\mu,h)\to (1,0)}\bar{x}(\mu,h)=T(\liminf_{(\mu,h)\to (1,0)}\bar{x}(\mu,h);1,0)
\end{align*}
hence both $\limsup_{(\mu,h)\to (1,0)}\bar{x}(\mu,h)$ and $\liminf_{(\mu,h)\to (1,0)}\bar{x}(\mu,h)$ satisfy the consistency equation:
\begin{align*}
    m=\mathbb{E}_z\tanh(z\sqrt{m}+m)
\end{align*}whose solution $m=0$ is unique, since the derivative of $T(m;1,0)$ is $\leq1$ and equality holds only at $m=0$. We conclude that there exists
\begin{align}
    \lim_{(\mu,h)\to (1,0)}\bar{x}(\mu,h)=0=\bar{x}(1,0)\;.
\end{align}

\emph{3.}
We first notice that 
\begin{align}\label{N_identity_gas}
    \mathbb{E}\tanh^2\left(z\sqrt{Q}+Q\right)=\mathbb{E}\tanh\left(z\sqrt{Q}+Q\right)\,,\quad Q\geq0\,,z\sim\mathcal{N}(0,1) 
\end{align} which simply follows from the third relation in \eqref{properties_psi} and the identity \eqref{N_identity_1}. Indeed, the quantity in \eqref{N_identity_gas} is nothing but the quenched average magnetization of a free system.
Now, by computing the first and second derivatives of the map $T(x;\mu,0)$ and using \eqref{N_identity_gas} we get:
\begin{align*}
    &T'(0;\mu,0)=\mu\,,\quad T''(0;\mu,0)=-2\mu^2\\
    &\bar{x}=\mu \bar{x}-\mu^2 \bar{x}^2(1+o(1))\quad\Rightarrow\quad
    \bar{x}=(1+o(1))\left(\frac{\mu-1}{\mu^2}\right)\;,
\end{align*}which implies that, in proximity of $\mu=1$, the magnetization goes to $0$ with a critical exponent $\beta=1$ (not to be confused with inverse absolute temperature) and with slope $1$.

\emph{4.} An analogous expansion of $T$ yields:
\begin{align*}
    \bar{x}=T(\bar{x};\mu,\lambda(\mu-1))=\mu\bar{x}+\lambda(\mu-1)-(\mu^2\bar{x}^2+o(\mu-1))(1+o(1))
\end{align*}
which in turn entails:
\begin{align*}
    \bar{x}^2=\frac{\lambda(\mu-1)}{\mu^2}(1+o(1))\;.
\end{align*}

\emph{5.} Here by $o(1)$ we mean a quantity that approaches $0$ as $h\to 0_+$.
As in the previous steps:
\begin{align*}
    \bar{x}=T(\bar{x};1,h)=\bar{x}+h-(\bar{x}^2+o(h))(1+o(1))\;,
\end{align*}then we get:
\begin{align*}
    \bar{x}^2=h(1+o(1))\quad\Rightarrow\quad\delta=2\;.
\end{align*}
\end{proof}

\printbibliography
\end{document}